\documentclass[runningheads]{llncs}

\usepackage{amsmath,amssymb,mathtools}
\usepackage{booktabs}
\usepackage{microtype}
\usepackage[hidelinks]{hyperref}

\newcommand{\cost}{\operatorname{cost}}
\newcommand{\imb}{\operatorname{imb}}
\newcommand{\DN}{\textnormal{\textsc{Distinct-N3DM}}}
\newcommand{\UBM}{\textnormal{\textsc{Unrestricted Balanced Mobiles}}}
\newcommand{\abs}[1]{\left|#1\right|}

\title{Strong NP-Completeness of Unrestricted Balanced Mobiles}
\titlerunning{Strong NP-Completeness of Unrestricted Balanced Mobiles}
\author{Andrei Popa and Alexandru Popa}
\authorrunning{A. Popa and A. Popa}
\institute{Faculty of Mathematics and Computer Science, University of Bucharest}

\begin{document}
\maketitle

\begin{abstract}
A mobile is a rooted full binary tree whose leaves carry positive integer
weights.  The imbalance of an internal node is the absolute difference
between the total weights of its two child subtrees, and the cost of the
mobile is the sum of these imbalances.  In the unrestricted
\emph{Balanced Mobiles} problem, only the multiset of leaf weights is given:
both the tree topology and the placement of the weights must be chosen so as
to minimize the cost.  The computational complexity of this unrestricted
variant has remained open, although the variant with a prescribed topology
is strongly NP-hard.  We close this gap by proving that the decision version
of unrestricted Balanced Mobiles is strongly NP-complete.  Our reduction
from Numerical 3-Dimensional Matching with Distinct Integers uses three
widely separated numerical scales.  Tight telescoping bounds force every
threshold-achieving mobile into a canonical hierarchy, after which
pairwise distinctness of the source integers collapses the hierarchy to
single triples from which a valid numerical matching can be recovered.
\end{abstract}

\section{Introduction}\label{sec:intro}

Quantifying how balanced a rooted tree  is a classical problem in
phylogenetics and in the combinatorial study of tree shapes.  Two of the
best-known measures are the Sackin index~\cite{Sackin1972}, based on leaf
depths, and the Colless index~\cite{Colless1982}, which for a rooted binary
tree sums, over all internal nodes, the absolute difference between the
numbers of leaves in the two child subtrees.  The extremal and structural
properties of these indices have generated a substantial literature; see,
for example, the characterization of Colless-minimal trees by Coronado et
al.~\cite{CoronadoEtAl2020} and the recent survey of Fischer et
al.~\cite{FischerEtAl2023}.  More general balance indices have also been
proposed to accommodate nonbinary trees and nonuniform node sizes
\cite{LemantEtAl2022}.

Balanced mobiles turn this structural measure into an optimization problem.
A \emph{mobile} is a rooted full binary tree whose leaves are assigned
positive integer weights.  At an internal node, the number of descendant
leaves in the Colless index is replaced by the total weight of the descendant
leaves.  Thus the cost of a mobile is the sum, over its internal nodes, of the
absolute difference between the total weights of the two child subtrees.
Given a multiset of weights, the unrestricted \emph{Balanced Mobiles}
problem asks us to choose both a full binary tree and a bijective assignment
of the given weights to its leaves so that this total cost is minimized.  We
write \UBM{} for the associated decision problem, in which a threshold is
part of the input.

Hamoudi, Laplante, and Mantaci~\cite{HamoudiEtAl2015} introduced balanced
mobiles as a weighted extension of the Colless index and studied their
connections with phylogenetic balance and Huffman-like optimization.  They
obtained algorithms for restricted cases and left the complexity of the
general optimization problem open.  A subsequent result of Ardevol
Martinez, Rizzi, and Sikora~\cite{ArdevolEtAl2023} established strong
NP-hardness for a different variant in which the full binary tree topology is
part of the input and only the assignment of weights to leaves is optimized.
Their reduction can exploit the prescribed topology to keep designated
weights in the same local structure.  That mechanism is unavailable in the
unrestricted problem: a candidate solution may reorganize all weights into a
completely different tree.  Proving hardness therefore requires controlling
the topology indirectly through the numerical values of the weights.

\paragraph{Our contribution.}
We resolve the complexity of the unrestricted problem.

\begin{theorem}\label{thm:main}
The decision version of \UBM{} is strongly NP-complete.  Consequently, the
corresponding optimization problem is strongly NP-hard.
\end{theorem}

The reduction is from Numerical 3-Dimensional Matching with Distinct
Integers (\DN{}), a strongly NP-hard problem even under the distinctness
restriction; in fact, the distinct version is strongly ASP-complete
\cite{MITHardness2024}.  The main technical difficulty is not creating a
low-cost mobile from a valid matching, but proving the converse when the tree
topology is unconstrained.  We address this by encoding several structural
coordinates at widely separated numerical scales.  Any mobile meeting the
threshold must make a sequence of lower bounds tight, and the equality cases
of those bounds force a canonical local structure.  The full normal-form
argument is developed after the construction, where all required notation is
available.

\paragraph{Organization.}
Section~\ref{sec:problem} defines the problem formally.
Section~\ref{sec:source} introduces and normalizes the source problem, and
Section~\ref{sec:reduction} gives the reduction together with a roadmap of
the soundness argument.  Section~\ref{sec:correctness} proves completeness
and soundness, while Section~\ref{sec:hardness} concludes strong
NP-completeness.  Repetitive second-scale arguments, routine estimates, and
the size calculation are deferred to the appendix.

\section{Problem definition and preliminaries}\label{sec:problem}

\begin{definition}[Mobile and cost]
A \emph{mobile} is a rooted full binary tree $T$ together with a positive integer weight on each leaf.  For a node $v$, let $W(v)$ denote the sum of all leaf weights in the subtree rooted at $v$.  If $v$ is internal with children $v_L,v_R$, define
\[
  \imb(v)=\abs{W(v_L)-W(v_R)}.
\]
The total cost is
\[
  \cost(T)=\sum_{v\text{ internal}}\imb(v).
\]
\end{definition}

\begin{problem}[\UBM]\label{prob:ubm}
\textbf{Input:} Positive integers $w_1,\dots,w_N$ and an integer $K\ge 0$.

\textbf{Question:} Does there exist a rooted full binary tree with $N$ leaves, with the given weights assigned bijectively to its leaves, whose total cost is at most $K$?
\end{problem}

\section{Source problem and normalization}\label{sec:source}

We reduce from Numerical 3-Dimensional Matching with Distinct Integers (\DN{}).  The distinct-integers restriction is known to preserve strong NP-hardness, and MIT Hardness Group et al.~\cite{MITHardness2024} strengthen this to strong ASP-completeness.  We use only strong NP-hardness.

\begin{problem}[Distinct-N3DM]\label{prob:dn3dm}
\textbf{Input:} Three sets
\[
 A^0=\{\alpha_1,\dots,\alpha_n\},\qquad
 B^0=\{\beta_1,\dots,\beta_n\},\qquad
 C^0=\{\gamma_1,\dots,\gamma_n\}
\]
of positive integers, with all $3n$ integers pairwise distinct, and an integer target $t$ satisfying
\[
  \sum_{i=1}^n\alpha_i+\sum_{j=1}^n\beta_j+\sum_{k=1}^n\gamma_k=nt.
\]
\textbf{Question:} Can $A^0\cup B^0\cup C^0$ be partitioned into $n$ triples, each containing exactly one element from each set and each summing to $t$?
\end{problem}

We may assume $n\ge 2$; the finitely many smaller instances can be handled directly.

Let
\[
 U=1+\max\!\left(\{t\}\cup A^0\cup B^0\cup C^0\right),
 \qquad M=3U+1,
\]
and define
\[
 a_i=10M+\alpha_i,\qquad
 b_j=4M+\beta_j,\qquad
 c_k=M+\gamma_k,
\]
with target
\[
 H=15M+t.
\]
Let $A=\{a_1,\dots,a_n\}$ and similarly $B,C$.

\begin{lemma}[Normalization]\label{lem:normalize}
The normalized instance satisfies:
\begin{align}
 a_i+b_j+c_k=H
 &\iff \alpha_i+\beta_j+\gamma_k=t, \label{eq:norm-equiv}\\
 \sum_i a_i+\sum_jb_j+\sum_kc_k&=nH, \label{eq:norm-total}\\
 b_j&>c_k \qquad\text{for all }j,k, \label{eq:bgtc}\\
 a_i&>b_j+c_k \qquad\text{for all }i,j,k, \label{eq:agtbc}\\
 0<a_i,b_j,c_k&<H. \label{eq:lessH}
\end{align}
Moreover, the elements within each of $A,B,C$ remain pairwise distinct.
\end{lemma}

The proof is deferred to Appendix~\ref{app:normalization}.

Two aggregate consequences will be used repeatedly.  From~\eqref{eq:bgtc}, after pairing the $n$ elements of $B$ and $C$ arbitrarily,
\begin{equation}\label{eq:sumbgtc}
 \sum_jb_j>\sum_kc_k.
\end{equation}
Similarly, from~\eqref{eq:agtbc},
\begin{equation}\label{eq:sumagtbc}
 \sum_i a_i>\sum_jb_j+\sum_kc_k.
\end{equation}

\section{The reduction}\label{sec:reduction}

Let $P$ be the least power of two with $P\ge n$, and let
\[
 q=P-n.
\]
Thus $P<2n$ for $n>1$ and $q\ge0$.  Choose
\begin{equation}\label{eq:Qdef}
 Q=100P^3(H+1).
\end{equation}
For each normalized source element create a leaf of the following weight:
\begin{align*}
 C_k^*&=Q^3+c_k,\\
 B_j^*&=Q^3+Q^2+b_j,\\
 A_i^*&=2Q^3+Q^2+Q+a_i.
\end{align*}
In addition, create $q$ padding leaves, each of weight
\[
 T^*=4Q^3+2Q^2+Q+H.
\]
The threshold is
\begin{equation}\label{eq:Ldef}
 L=nQ^2+nQ+\sum_i a_i-2\sum_k c_k.
\end{equation}
By~\eqref{eq:sumbgtc}--\eqref{eq:sumagtbc}, $L>0$.

It is useful to associate to each leaf four additive coordinates $(m,d,e,r)$ by writing its constructed weight as
\[
 Q^3m+Q^2d+Qe+r.
\]
These are defined by the construction rather than inferred from an arbitrary base-$Q$ expansion.  They are:
\begin{center}
\begin{tabular}{c@{\qquad}cccc}
\toprule
leaf type & $m$ & $d$ & $e$ & $r$\\
\midrule
$C_k^*$ & $1$ & $0$ & $0$ & $c_k$\\
$B_j^*$ & $1$ & $1$ & $0$ & $b_j$\\
$A_i^*$ & $2$ & $1$ & $1$ & $a_i$\\
$T^*$   & $4$ & $2$ & $1$ & $H$\\
\bottomrule
\end{tabular}
\end{center}
For a subtree, each coordinate denotes the sum of that coordinate over its descendant leaves.  Over the full instance,
\begin{equation}\label{eq:coordtotals}
 m_{\mathrm{tot}}=4P,\qquad
 d_{\mathrm{tot}}=2P,\qquad
 e_{\mathrm{tot}}=P,\qquad
 r_{\mathrm{tot}}=PH.
\end{equation}
The number of leaves is
\begin{equation}\label{eq:Nleaves}
 N=3n+q=2n+P\le3P.
\end{equation}

For an internal node $v$ and an additive coordinate $x$, write
\[
 \Delta x_v=x(v_L)-x(v_R).
\]
Its imbalance is therefore
\begin{equation}\label{eq:nodecostcoords}
 \imb(v)=\abs{Q^3\Delta m_v+Q^2\Delta d_v+Q\Delta e_v+\Delta r_v}.
\end{equation}

\subsection{Scale separation}

\begin{lemma}[Scale inequalities]\label{lem:scale}
For $Q$ as in~\eqref{eq:Qdef},
\begin{align}
 Q&>(3P^2+P)H, \label{eq:scale1}\\
 Q^2&>(3P^2+P)(Q+H), \label{eq:scale2}\\
 Q^3&>3PQ^2+2PQ+2PH. \label{eq:scale3}
\end{align}
In particular, $Q>PH$.
\end{lemma}

The proof is deferred to Appendix~\ref{app:estimates}.

\begin{lemma}[Variation bound]\label{lem:variation}
Assign a nonnegative value $x(\ell)$ to each leaf $\ell$, and extend $x$
additively to every subtree: $x(S)$ is the sum of the values on the leaves of
$S$.  Let $X=x(\mathrm{root})$.  If $X>0$, then
\[
 \sum_{v\text{ internal}}\abs{\Delta x_v}<3PX.
\]
If $X=0$, the left-hand side is $0$.  In particular,
\begin{equation}\label{eq:variation-er}
 \sum_v\abs{\Delta e_v}<3P^2,
 \qquad
 \sum_v\abs{\Delta r_v}<3P^2H.
\end{equation}
\end{lemma}

The proof is deferred to Appendix~\ref{app:estimates}.

\subsection{Roadmap of the soundness argument}\label{sec:roadmap}

The completeness direction uses the matching triples directly and is
straightforward.  Soundness is the central part of the proof, because the
constructed instance does not prescribe a topology.  We briefly describe the
logic before giving the details.

First, the $Q^3$ coordinate acts as a structural mass.  Scale separation
implies that every internal node of a mobile of cost at most $L$ must split
this mass equally.  Once this is known, the $Q^2$ coordinate can be analyzed
through the additive potential $\delta=2d-m$.  A telescoping lower bound
shows that the total $d$-variation has a forced minimum.  Reaching the target
cost makes this bound tight at every node; its equality condition implies
that every nonzero $d$-difference occurs at a node separating a pure block of
$B^*$ leaves from an equally large pure block of $C^*$ leaves.  We call such
a node a $BC$ junction.

We then contract balanced $BC$ blocks and repeat the same argument at the
$Q$ scale with a second potential.  Tightness now forces every nonzero
$e$-difference to join a pure block of $A^*$ leaves to an equally large
balanced $BC$ block.  These two families of forced junctions already account
for the entire threshold $L$.  Consequently, every remaining internal node
must have zero residual imbalance.

Finally, distinctness of the normalized source integers turns this global
normal form into local triples: any cancellation block of size at least two
contains a cherry of two leaves of the same type, and the zero-residual
condition would force their residual weights to be equal.  Thus all forced
junctions have size one.  After contracting the resulting $ABC$ triples, the
upper part of the tree consists of objects with identical high-order
coordinates; equal structural mass and zero residual imbalance force all
triple sums to equal the target $H$.  The original numerical matching can
then be read off from these triples.

\section{Correctness of the reduction}\label{sec:correctness}

We first prove completeness and then soundness.  The soundness proof follows
the three numerical scales of the construction; rather than introducing a new
subsection for each lemma, we use short transitions between the successive
steps.

\subsection{Completeness}

\begin{lemma}\label{lem:yes}
If the \DN{} instance is a yes-instance, then the constructed \UBM{} instance has a mobile of cost exactly $L$.
\end{lemma}

\begin{proof}
Take a valid matching.  For each matched triple $(a_i,b_j,c_k)$, first join $B_j^*$ and $C_k^*$.  By~\eqref{eq:bgtc}, this node costs
\[
 Q^2+b_j-c_k.
\]
Join its parent to $A_i^*$.  The $Q^3$ and $Q^2$ terms cancel, and~\eqref{eq:agtbc} gives cost
\[
 Q+a_i-b_j-c_k.
\]
Thus the three-leaf gadget costs
\begin{equation}\label{eq:gadgetcost}
 Q^2+Q+a_i-2c_k.
\end{equation}
Its total weight is
\[
 4Q^3+2Q^2+Q+(a_i+b_j+c_k)=T^*.
\]
After forming all $n$ gadgets, we have those $n$ equal-weight gadget roots together with $q=P-n$ padding leaves of the same weight $T^*$.  Arrange these $P$ objects as the leaves of a perfect binary tree.  Every imbalance above the gadgets is zero.  Summing~\eqref{eq:gadgetcost} over the matching gives exactly
\[
 nQ^2+nQ+\sum_i a_i-2\sum_kc_k=L.
\]
\qed
\end{proof}

\subsection{Soundness}

We now prove the converse.  Fix a mobile $\mathcal T$ of cost at most $L$.
We begin with a general telescoping inequality that will be used at both lower
scales.

\begin{lemma}[Telescoping potential]\label{lem:telescoping}
Let $\phi$ be any real-valued additive function on subtrees.  Then
\begin{equation}\label{eq:telescoping}
 \sum_{v\text{ internal}}\abs{\phi(v_L)-\phi(v_R)}
 \ge
 \sum_{\ell\text{ leaf}}\abs{\phi(\ell)}-\abs{\phi(\mathrm{root})}.
\end{equation}
If equality holds in~\eqref{eq:telescoping}, then at every internal node $v$,
\begin{equation}\label{eq:equalabs}
 \abs{\phi(v_L)}=\abs{\phi(v_R)}.
\end{equation}
\end{lemma}

\begin{proof}
Let $v$ be an internal node with children $u,w$.  Additivity gives
$\phi(v)=\phi(u)+\phi(w)$.  Define its local gap by
\[
 g(v)=|\phi(u)-\phi(w)|+|\phi(v)|-|\phi(u)|-|\phi(w)|.
\]
For arbitrary real $x,y$,
\[
 |x-y|+|x+y|=2\max\{|x|,|y|\}\ge |x|+|y|,
\]
so $g(v)\ge0$.  Summing these gaps over all internal nodes telescopes: every
non-root internal term $|\phi(v)|$ appears once with a plus sign at $v$ and
once with a minus sign at its parent.  Hence
\[
 \sum_{v\text{ internal}}|\phi(v_L)-\phi(v_R)|
 -\sum_{\ell\text{ leaf}}|\phi(\ell)|+|\phi(\mathrm{root})|
 =\sum_{v\text{ internal}}g(v)\ge0,
\]
which is~\eqref{eq:telescoping}.

If equality holds globally, then every nonnegative local gap is zero.  From
the displayed identity for $|x-y|+|x+y|$, this happens at $v$ exactly when
$|\phi(v_L)|=|\phi(v_R)|$, giving~\eqref{eq:equalabs}.
\qed
\end{proof}

The highest-order coordinate first forces every split to preserve structural mass.

\begin{lemma}[Equal structural mass]\label{lem:mzero}
If a mobile $\mathcal T$ for the constructed instance satisfies $\cost(\mathcal T)\le L$, then
\[
 \Delta m_v=0
 \qquad\text{for every internal node }v.
\]
\end{lemma}

\begin{proof}
Suppose $\Delta m_v\ne0$ for some internal node $v$.  Since $\Delta m_v$ is integral, $\abs{\Delta m_v}\ge1$.  Using~\eqref{eq:coordtotals},
\[
 \abs{\Delta d_v}\le2P,\qquad
 \abs{\Delta e_v}\le P,\qquad
 \abs{\Delta r_v}\le PH.
\]
Thus~\eqref{eq:nodecostcoords} gives
\[
 \imb(v)\ge Q^3-2PQ^2-PQ-PH.
\]
On the other hand, by~\eqref{eq:lessH},
\[
 L<nQ^2+nQ+PH\le PQ^2+PQ+PH.
\]
Inequality~\eqref{eq:scale3} is equivalent to
\[
 Q^3-2PQ^2-PQ-PH>PQ^2+PQ+PH,
\]
so $\imb(v)>L$, contradicting $\cost(\mathcal T)\le L$.
\qed
\end{proof}

For the remainder of the soundness proof, fix a mobile $\mathcal T$ satisfying $\cost(\mathcal T)\le L$.  By Lemma~\ref{lem:mzero}, every internal node has equal structural mass on its two sides, and hence its cost simplifies to
\begin{equation}\label{eq:costafterM}
 \imb(v)=\abs{Q^2\Delta d_v+Q\Delta e_v+\Delta r_v}.
\end{equation}

We next analyze the $Q^2$ coordinate.  The potential $\delta=2d-m$ isolates the discrepancy between $B^*$ and $C^*$ leaves.

Define the additive potential
\begin{equation}\label{eq:deltadef}
 \delta=2d-m.
\end{equation}
For individual leaves,
\begin{equation}\label{eq:deltaleaves}
 \delta(C_k^*)=-1,\qquad
 \delta(B_j^*)=1,\qquad
 \delta(A_i^*)=\delta(T^*)=0.
\end{equation}
The root has $\delta=0$, and Lemma~\ref{lem:mzero} gives
\begin{equation}\label{eq:deltad}
 \Delta\delta_v=2\Delta d_v.
\end{equation}

\begin{lemma}\label{lem:dexact}
\[
 \sum_v\abs{\Delta d_v}=n.
\]
Consequently,
\begin{equation}\label{eq:deltaequal}
 \abs{\delta(v_L)}=\abs{\delta(v_R)}
 \qquad\text{at every internal node }v.
\end{equation}
\end{lemma}

\begin{proof}
Applying Lemma~\ref{lem:telescoping} to $\delta$, using~\eqref{eq:deltaleaves} and~\eqref{eq:deltad}, gives
\[
 2\sum_v\abs{\Delta d_v}
 =\sum_v\abs{\Delta\delta_v}
 \ge 2n,
\]
so $\sum_v\abs{\Delta d_v}\ge n$.  Since $d$ is integer-valued, $\sum_v\abs{\Delta d_v}$ is an integer.

Suppose instead that it is at least $n+1$.  By~\eqref{eq:costafterM}, the triangle inequality, and Lemma~\ref{lem:variation},
\begin{align*}
 \cost(\mathcal T)
 &\ge Q^2\sum_v\abs{\Delta d_v}
      -Q\sum_v\abs{\Delta e_v}
      -\sum_v\abs{\Delta r_v}\\
 &>(n+1)Q^2-3P^2Q-3P^2H.
\end{align*}
By~\eqref{eq:scale2},
\[
 Q^2-3P^2Q-3P^2H>PQ+PH.
\]
Hence
\[
 \cost(\mathcal T)>nQ^2+PQ+PH>L,
\]
a contradiction.  Thus equality holds in the telescoping bound, and Lemma~\ref{lem:telescoping} yields~\eqref{eq:deltaequal}.
\qed
\end{proof}

\begin{lemma}[Sign purity for $\delta$]\label{lem:deltapurity}
If a subtree $S$ has $\delta(S)>0$, then every leaf in $S$ is of type $B^*$.  If $\delta(S)<0$, then every leaf in $S$ is of type $C^*$.
\end{lemma}

\begin{proof}
Induct on the number of leaves in $S$.  The leaf case follows from~\eqref{eq:deltaleaves}.  If $S$ is internal and $\delta(S)>0$, let $x=\delta(S_L)$ and $y=\delta(S_R)$.  Equation~\eqref{eq:deltaequal} gives $\abs{x}=\abs{y}$, while $x+y>0$.  Hence $x=y>0$, and the induction hypothesis applies to both children.  The negative case is symmetric.
\qed
\end{proof}

\begin{definition}[$BC$ junction]
An internal node is a \emph{$BC$ junction of size $s$} if one child consists of exactly $s$ leaves of type $B^*$ and the other consists of exactly $s$ leaves of type $C^*$.
\end{definition}

\begin{lemma}[Exact $BC$ normal form]\label{lem:bcjunction}
An internal node has $\Delta d_v\ne0$ if and only if it is a $BC$ junction.  At a size-$s$ $BC$ junction,
\[
 \Delta e_v=0
\]
and its cost is
\begin{equation}\label{eq:bccost}
 sQ^2+\sum_{B\text{ side}}b_j-\sum_{C\text{ side}}c_k.
\end{equation}
Moreover, the $BC$ junctions partition all $B^*$ and $C^*$ leaves, and their total cost is exactly
\begin{equation}\label{eq:totalbccost}
 C_{BC}=nQ^2+\sum_jb_j-\sum_kc_k.
\end{equation}
\end{lemma}

\begin{proof}
If $\Delta d_v\ne0$, then by~\eqref{eq:deltad} the child $\delta$-values differ.  Equation~\eqref{eq:deltaequal} forces them to be $s$ and $-s$ for some $s\ge1$.  Lemma~\ref{lem:deltapurity} then implies that the positive child contains exactly $s$ $B^*$ leaves and the negative child exactly $s$ $C^*$ leaves.  The converse is immediate.  Both sides have $e=0$, proving $\Delta e_v=0$.

To prove the partition claim, fix a $B^*$ leaf $\ell$ and follow its unique
path to the root.  We have $\delta(\ell)=1$ and $\delta(\mathrm{root})=0$, so
there is a first node $v$ on this path with $\delta(v)=0$.  Let $u$ be the
child of $v$ containing $\ell$.  By minimality of $v$, all nodes from $\ell$
through $u$ have nonzero $\delta$.  Equation~\eqref{eq:deltaequal} implies
that a nonzero parent cannot change sign: if its children have equal absolute
$\delta$-values and their sum is nonzero, those values have the same sign.
Thus $\delta(u)>0$, and Lemma~\ref{lem:deltapurity} implies that every leaf
below $u$ is a $B^*$ leaf.  Since $\delta(v)=0$, the sibling $w$ of $u$ has
$\delta(w)=-\delta(u)<0$; hence every leaf below $w$ is a $C^*$ leaf.  The
magnitudes are equal, so the two pure sides contain the same number of leaves.
Therefore $v$ is a $BC$ junction.

The first zero ancestor of $\ell$ is unique.  All $B^*$ leaves in the same
positive child $u$ have this same first zero ancestor, and no $B^*$ leaf
outside $u$ does.  Consequently the positive children of the $BC$ junctions
are disjoint and cover all $B^*$ leaves.  The symmetric argument for $C^*$
leaves shows that the negative children are disjoint and cover all $C^*$
leaves.  Hence the $BC$ junctions partition all $B^*$ and $C^*$ leaves.

At a size-$s$ junction the high coordinates differ only in $d$, and~\eqref{eq:bgtc} makes the residual difference positive, yielding~\eqref{eq:bccost}.  Since every $B^*$ and $C^*$ leaf occurs in exactly one such junction and the sizes sum to $n$, summing~\eqref{eq:bccost} gives~\eqref{eq:totalbccost}.
\qed
\end{proof}

After the $B^*$ and $C^*$ structure is fixed, the $Q$ scale is analogous.  We contract the balanced $BC$ blocks and use a second potential.  To avoid repeating the same telescoping argument, the proofs of the three lemmas below are deferred to Appendix~\ref{app:qscale}.

Contract each maximal subtree containing only $B^*,C^*$ leaves and having total $\delta=0$ into a single meta-leaf.  By Lemma~\ref{lem:bcjunction}, these contracted blocks are disjoint and collectively contain all $B^*,C^*$ leaves.  A block containing $s$ $B^*$ leaves and $s$ $C^*$ leaves has $m=2s$ and $e=0$.

On the contracted tree define
\begin{equation}\label{eq:epsdef}
 \varepsilon=2e-\frac{m}{2}.
\end{equation}
Every contracted leaf has even $m$, and its $\varepsilon$-value is
\begin{equation}\label{eq:epsleaves}
 \varepsilon(A_i^*)=1,\qquad
 \varepsilon(T^*)=0,\qquad
 \varepsilon(BC_s)=-s.
\end{equation}
The root has $\varepsilon=0$.  Since $\Delta m=0$ throughout the original tree,
\begin{equation}\label{eq:eps-e}
 \Delta\varepsilon_v=2\Delta e_v
\end{equation}
for every internal node of the contracted tree.  Nodes removed by the contraction have $e=0$ throughout, so they contribute no $e$-variation.  Consequently, the total $e$-variation is the same whether it is computed in the original tree or in the contracted tree.

\begin{lemma}\label{lem:eexact}
\[
 \sum_v\abs{\Delta e_v}=n.
\]
Consequently, at every internal node of the contracted tree,
\begin{equation}\label{eq:epsequal}
 \abs{\varepsilon(v_L)}=\abs{\varepsilon(v_R)}.
\end{equation}
\end{lemma}

The proof is deferred to Appendix~\ref{app:qscale}.

\begin{lemma}[Sign purity for $\varepsilon$]\label{lem:epspurity}
In the contracted tree, a subtree with positive $\varepsilon$ consists entirely of $A^*$ leaves; a subtree with negative $\varepsilon$ consists entirely of balanced $BC$ meta-leaves.
\end{lemma}

The proof is deferred to Appendix~\ref{app:qscale}.

\begin{definition}[$A/BC$ junction]
An internal node of the contracted tree is an \emph{$A/BC$ junction of size $s$} if one child consists of exactly $s$ leaves of type $A^*$ and the other consists of balanced $BC$ blocks containing, in total, exactly $s$ $B^*$ leaves and $s$ $C^*$ leaves.
\end{definition}

\begin{lemma}[Exact $A/BC$ normal form]\label{lem:abcjunction}
An internal node of the contracted tree has $\Delta e_v\ne0$ if and only if it is an $A/BC$ junction.  Every such junction has $\Delta d_v=0$.  The $A/BC$ junctions partition all $A^*$ leaves and all balanced $BC$ blocks, and their total cost is exactly
\begin{equation}\label{eq:totalabccost}
 C_{A/BC}=nQ+\sum_i a_i-\sum_jb_j-\sum_kc_k.
\end{equation}
\end{lemma}

The proof is deferred to Appendix~\ref{app:qscale}.

The two scale arguments combine into an exact accounting identity.

\begin{theorem}[Normal form]\label{thm:normalform}
Every internal node of $\mathcal T$ belongs to exactly one of the following three classes:
\begin{enumerate}
 \item[(i)] a $BC$ junction, for which $\Delta m=0$, $\Delta d\ne0$, and $\Delta e=0$;
 \item[(ii)] an $A/BC$ junction, for which $\Delta m=0$, $\Delta d=0$, and $\Delta e\ne0$;
 \item[(iii)] a \emph{residual-only node}, for which $\Delta m=\Delta d=\Delta e=0$.
\end{enumerate}
If $R$ is the set of residual-only nodes, then
\begin{equation}\label{eq:normalidentity}
 \cost(\mathcal T)=L+\sum_{v\in R}\abs{\Delta r_v}.
\end{equation}
\end{theorem}

\begin{proof}
Lemma~\ref{lem:mzero} gives $\Delta m=0$ everywhere.  If $\Delta d\ne0$, Lemma~\ref{lem:bcjunction} gives class~(i), including $\Delta e=0$.  Otherwise, if $\Delta e\ne0$, the node cannot lie inside a contracted $BC$ block because $e=0$ throughout such a block.  It therefore survives the contraction, and Lemma~\ref{lem:abcjunction} gives class~(ii).  If neither occurs, the node is in class~(iii).  The classes are therefore exhaustive and mutually exclusive.

By~\eqref{eq:totalbccost} and~\eqref{eq:totalabccost}, the total contribution of classes~(i) and~(ii) is
\begin{align*}
 C_{BC}+C_{A/BC}
 &=\left(nQ^2+\sum b-\sum c\right)
   +\left(nQ+\sum a-\sum b-\sum c\right)\\
 &=nQ^2+nQ+\sum a-2\sum c\\
 &=L.
\end{align*}
Every node in class~(iii) has cost exactly $\abs{\Delta r_v}$, proving~\eqref{eq:normalidentity}.
\qed
\end{proof}

\begin{corollary}\label{cor:reszero}
If $\cost(\mathcal T)\le L$, then $\cost(\mathcal T)=L$ and
\[
 \Delta r_v=0
 \qquad\text{for every residual-only node }v.
\]
\end{corollary}

\begin{proof}
Immediate from~\eqref{eq:normalidentity} and nonnegativity.
\qed
\end{proof}

We now use pairwise distinctness to collapse every cancellation block to a single triple.

\begin{lemma}[Unit-size cancellation gadgets]\label{lem:unitsize}
Every $BC$ junction and every $A/BC$ junction has size one.
\end{lemma}

\begin{proof}
Suppose a $BC$ junction has size $s\ge2$.  Its pure-$B^*$ child is a full binary tree with at least two leaves, so it contains a cherry: an internal node whose children are two leaves $B_i^*,B_j^*$.  At their parent the two children have identical $(m,d,e)$ coordinates $(1,1,0)$, so that parent is residual-only.  Corollary~\ref{cor:reszero} therefore implies
\[
 b_i=b_j,
\]
contradicting distinctness within $B$.  Hence every $BC$ junction has size one.  The same argument using a cherry of two $A^*$ leaves, whose common high coordinates are $(2,1,1)$, shows that every $A/BC$ junction has size one.
\qed
\end{proof}

Consider a maximal balanced $BC$ block containing $s$ leaves of type $B^*$ and $s$ leaves of type $C^*$.  It has $\varepsilon=-s$ in the contracted tree.  Such a block belongs to an $A/BC$ junction of size $s$ by Lemma~\ref{lem:abcjunction}.  Since every $A/BC$ junction has size one, every maximal balanced $BC$ block contains exactly one $B^*$ leaf and one $C^*$ leaf.  Its size-one $A/BC$ junction then attaches exactly one $A^*$ leaf.  Therefore the $3n$ nonpadding source leaves are partitioned into $n$ disjoint three-leaf components, each containing exactly one $A_i^*$, one $B_j^*$, and one $C_k^*$.  The total weight of such a component is
\begin{equation}\label{eq:tripleweight}
 4Q^3+2Q^2+Q+s,
 \qquad s=a_i+b_j+c_k.
\end{equation}

It remains to show that every resulting triple has normalized sum $H$.

Contract each of the $n$ three-leaf components into one meta-leaf.  Together with the $q=P-n$ padding leaves, there are exactly $P$ meta-leaves.  Every meta-leaf has the same high coordinates
\[
 (m,d,e)=(4,2,1).
\]
The residual of a triple meta-leaf is its sum $s$ from~\eqref{eq:tripleweight}; the residual of a padding leaf is $H$.

\begin{lemma}[Equal-count/equal-sum trees]\label{lem:equalleaves}
Let a rooted full binary tree have a real number on each leaf.  Suppose that at every internal node (i) its two children contain the same number of leaves, and (ii) the sums of the leaf numbers in its two child subtrees are equal.  Then all leaves carry the same number.
\end{lemma}

The proof is deferred to Appendix~\ref{app:equalleaves}.

\begin{lemma}\label{lem:triplesH}
Every three-leaf component obtained above satisfies
\[
 a_i+b_j+c_k=H.
\]
\end{lemma}

\begin{proof}
Every internal node remaining after the contraction is residual-only.  Hence Corollary~\ref{cor:reszero} says that its two child residual sums are equal.  Furthermore, Lemma~\ref{lem:mzero} gives equal $m$-mass on its two sides.  Since every meta-leaf has $m=4$, the two sides therefore contain the same number of meta-leaves.  Lemma~\ref{lem:equalleaves} implies that all $P$ meta-leaf residuals are equal.

Their total residual is, by~\eqref{eq:norm-total},
\[
 \sum_i a_i+\sum_jb_j+\sum_kc_k+qH
 =nH+qH=PH.
\]
Thus the common residual is $H$.  In particular, every triple meta-leaf has residual $a_i+b_j+c_k=H$.
\qed
\end{proof}

\begin{lemma}[Soundness]\label{lem:soundness}
If the constructed \UBM{} instance has a mobile of cost at most $L$, then the source \DN{} instance is a yes-instance.
\end{lemma}

\begin{proof}
By Lemma~\ref{lem:unitsize}, the source leaves are partitioned into $n$ triples containing exactly one $A$-, one $B$-, and one $C$-element.  Lemma~\ref{lem:triplesH} shows that each normalized triple sums to $H$.  By~\eqref{eq:norm-equiv}, the corresponding original triple sums to $t$.  Since every source element occurs exactly once, these triples form a valid numerical three-dimensional matching.
\qed
\end{proof}

\section{Strong NP-completeness}\label{sec:hardness}

\begin{proof}[Theorem~\ref{thm:main}]
Lemmas~\ref{lem:yes} and~\ref{lem:soundness} show that the source instance is
a yes-instance if and only if the constructed \UBM{} instance admits a mobile
of cost at most $L$.  The construction has polynomial size and, on the
strongly bounded source family, all generated numbers remain polynomially
bounded; the calculation is deferred to Appendix~\ref{app:size}.  Hence the
reduction proves strong NP-hardness.

Finally, \UBM{} is in NP.  A certificate is a rooted full binary tree with the
input weights assigned to its leaves.  It has $2N-1$ nodes, and all subtree
weights, imbalances, and the total cost can be computed bottom-up using
polynomially many bit operations.  Therefore \UBM{} is strongly NP-complete.
\qed
\end{proof}

\section{Conclusion}\label{sec:conclusion}

We have shown that unrestricted Balanced Mobiles is strongly NP-complete,
settling the complexity question left open by the original formulation.  The
unrestricted setting is substantially different from the prescribed-topology
variant: the reduction cannot assume that designated weights remain in fixed
local gadgets.  Instead, the proof derives the required topology from the
numerical instance itself.  Separated scales first eliminate structural
imbalance, and two tight potential inequalities then force every
threshold-achieving solution into a canonical hierarchy.  Distinct source
integers finally collapse the hierarchy to the triples of a numerical
three-dimensional matching.

The normal-form technique may be useful for other optimization problems on
weighted trees in which both the topology and the placement of weights are
part of the solution.  For Balanced Mobiles itself, natural next questions
include approximation guarantees, parameterized algorithms, and tractable
weight classes that exploit the relation with Huffman-like constructions
identified in the original work~\cite{HamoudiEtAl2015}.

\paragraph{Acknowledgments.}
The authors used generative AI, namely ChatGPT, to obtain these results. The model found the full solution of the problem and wrote a first draft. The model was also used to subsequently edit the paper into the final version presented here. The authors verified the results and assume full responsibility.

\clearpage
\appendix
\renewcommand{\theHsection}{appendix.\Alph{section}}
\section{Deferred proofs and estimates}\label{app:deferred}
\noindent\textit{The 14-page main paper, including the references, ends above. The material below is supplementary and contains deferred proofs and routine estimates.}

\subsection{Normalization}\label{app:normalization}
\begin{proof}[Lemma~\ref{lem:normalize}]
The equivalence~\eqref{eq:norm-equiv} follows by cancelling the common offset
$15M$, and~\eqref{eq:norm-total} follows from the total-sum condition of the
source instance.  Since every original number is smaller than $U$,
\[
 b_j-c_k>3M-U>0,
 \qquad
 a_i-b_j-c_k>5M-2U>0.
\]
The bounds in~\eqref{eq:lessH} follow from $M=3U+1$ and $t>0$.
Finally, adding a fixed offset within each of $A^0,B^0,C^0$ preserves
pairwise distinctness.
\qed
\end{proof}

\subsection{Scale separation and coordinate variation}\label{app:estimates}
\begin{proof}[Lemma~\ref{lem:scale}]
Because $P,H\ge1$, we have $3P^2+P\le4P^2$ and
$Q=100P^3(H+1)>100P^3H$, which gives~\eqref{eq:scale1}.  Also $Q>H$ and
$Q>8P^2$, hence
\[
 (3P^2+P)(Q+H)<4P^2(2Q)=8P^2Q<Q^2,
\]
proving~\eqref{eq:scale2}.  Finally $Q>8P$ and $Q>H$, so
\[
 3PQ^2+2PQ+2PH<7PQ^2<Q^3,
\]
which gives~\eqref{eq:scale3}.  In particular, $Q>PH$.
\qed
\end{proof}

\begin{proof}[Lemma~\ref{lem:variation}]
For every internal node $v$,
\[
 |\Delta x_v|\le x(v_L)+x(v_R)=x(v).
\]
Expanding $\sum_vx(v)$ leaf by leaf, each leaf contributes once for each of
its internal ancestors.  Every leaf has at most $N-1<3P$ internal ancestors
because $N\le3P$.  Thus, when $X>0$,
\[
 \sum_v|\Delta x_v|\le\sum_vx(v)\le(N-1)X<3PX.
\]
If $X=0$, nonnegativity forces $x$ to vanish everywhere.  Applying the
positive case to $e_{\rm tot}=P$ and $r_{\rm tot}=PH$ yields
\eqref{eq:variation-er}.
\qed
\end{proof}

\subsection{Equal-count/equal-sum trees}\label{app:equalleaves}
\begin{proof}[Lemma~\ref{lem:equalleaves}]
We induct on the number of leaves.  The one-leaf case is immediate.  Let the
root be internal.  Each of its two child subtrees again satisfies both
hypotheses, so by induction all leaves in the left child have a common value
$x$ and all leaves in the right child have a common value $y$.  The two
children contain the same positive number of leaves and have the same total
leaf sum.  Hence $x=y$, and all leaves of the tree have the same value.
\qed
\end{proof}

\subsection{Proofs for the $Q$-scale argument}\label{app:qscale}

\begin{proof}[Lemma~\ref{lem:eexact}]
The contracted leaves have total positive $\varepsilon$-mass $n$, contributed
by the $A^*$ leaves, and total negative mass $-n$, contributed by the balanced
$BC$ blocks.  Lemma~\ref{lem:telescoping} and~\eqref{eq:eps-e} give
$\sum_v|\Delta e_v|\ge n$.  This quantity is integral.  If it were at least
$n+1$, then Lemma~\ref{lem:bcjunction} and Lemma~\ref{lem:variation} would give
\begin{align*}
 \cost(\mathcal T)
 &\ge C_{BC}+Q\sum_v|\Delta e_v|-\sum_v|\Delta r_v|\\
 &>C_{BC}+(n+1)Q-3P^2H.
\end{align*}
By~\eqref{eq:scale1}, $Q-3P^2H>PH$, whereas~\eqref{eq:lessH} gives
$\sum_i a_i-\sum_jb_j-\sum_kc_k<PH$.  Thus
\[
 \cost(\mathcal T)>C_{BC}+nQ+
 \sum_i a_i-\sum_jb_j-\sum_kc_k=L,
\]
a contradiction.  Hence $\sum_v|\Delta e_v|=n$.  Equality in the telescoping
bound then yields~\eqref{eq:epsequal} at every internal node of the contracted
tree.
\qed
\end{proof}

\begin{proof}[Lemma~\ref{lem:epspurity}]
Induct on the number of contracted leaves.  The leaf case follows from
\eqref{eq:epsleaves}.  If $S$ is internal and $\varepsilon(S)>0$, let
$x=\varepsilon(S_L)$ and $y=\varepsilon(S_R)$.  By~\eqref{eq:epsequal},
$|x|=|y|$, while $x+y>0$ forces $x=y>0$.  Both children are therefore pure
$A^*$ by induction.  The negative case is symmetric.
\qed
\end{proof}

\begin{proof}[Lemma~\ref{lem:abcjunction}]
By~\eqref{eq:eps-e}, $\Delta e_v\ne0$ exactly when the two child
$\varepsilon$-values differ.  Equality~\eqref{eq:epsequal} then forces those
values to be $s$ and $-s$ for some integer $s\ge1$.  Lemma~\ref{lem:epspurity}
shows that the positive side contains exactly $s$ $A^*$ leaves, while the
negative side consists of balanced $BC$ blocks containing altogether $s$
$B^*$ and $s$ $C^*$ leaves.  Thus the node is an $A/BC$ junction; the converse
is immediate.  Both sides have total $d=s$, so $\Delta d_v=0$.

The first-zero-ancestor argument from Lemma~\ref{lem:bcjunction}, now applied
to $\varepsilon$ in the contracted tree, shows that these junctions are
disjoint and partition all $A^*$ leaves and all balanced $BC$ blocks.

At a size-$s$ junction the $Q^3$ and $Q^2$ terms cancel.  Moreover,
\eqref{eq:agtbc} implies
\[
 \sum_{A\text{ side}}a_i>
 \sum_{BC\text{ side}}b_j+\sum_{BC\text{ side}}c_k,
\]
so the cost is
\[
 sQ+\sum_{A\text{ side}}a_i-
 \sum_{BC\text{ side}}b_j-\sum_{BC\text{ side}}c_k.
\]
Summing over the partition gives~\eqref{eq:totalabccost}.
\qed
\end{proof}

\subsection{Size of the reduction}\label{app:size}
The source problem is strongly NP-hard~\cite{MITHardness2024}, so hardness
holds on a family in which every source integer is polynomially bounded in
$n$.  Then $U,M,H$ are polynomially bounded.  Since $P<2n$ for $n>1$,
$Q=100P^3(H+1)$ is also polynomially bounded.  The largest constructed weight
is $O(Q^3)$ and the instance has $N=2n+P=O(n)$ leaves.  Finally,
\[
 L=nQ^2+nQ+\sum_i a_i-2\sum_kc_k=O(nQ^2+nH),
\]
so the threshold is polynomially bounded.  Thus the construction is a
polynomial-time reduction that preserves strong NP-hardness.

\end{document}